\documentclass{PM}

\address[exner@ujf.cas.cz]{Pavel Exner$^{1,3}$}

\address[minakov.ilt@gmail.com]{Alexander Minakov$^{1,2}$}

\address[l.parnovski@ucl.ac.uk]{Leonid Parnovski$^4$}

\address[]{$^{(1)}$Doppler Institute for Mathematical Physics and Applied Mathematics, \\
Czech Technical University in Prague, B\v{r}ehov\'a 7, 11519
Prague, Czech Republic}

\address[]{$^{(2)}$Department of Physics, Faculty of Nuclear
Science and Physical Engineering, \\Czech Technical University in
Prague, Pohrani\v{c}ni 1288/1, 40501 D\v{e}\v{c}in, Czech \\
Republic}

\address[]{$^{(3)}$Department of Theoretical Physics, Nuclear Physics Institute ASCR, 25068
\v{R}e\v{z} near Prague, Czech Republic}

\address[]{$^{(4)}$Department of Mathematics, University College London \\
Gower Street, London WC1E 6BT, United Kingdom}

\newtheorem{theorem}{Theorem}[section]

\newtheorem{lemma}[theorem]{Lemma}

\theoremstyle{definition}

\newtheorem{example}[theorem]{Example}

\newtheorem{remarks}[theorem]{Remarks}

\newcommand{\e}{\mathrm{e}}
\renewcommand{\i}{\mathrm{i}}

\renewcommand{\d}{\mathrm{d}}
\newcommand{\ds}{\displaystyle}

\newcommand{\dsfrac}{\ds\frac}

\newcommand{\OO}{\mathcal{O}}
\newcommand{\R}{\mathbb{R}}
\newcommand{\Z}{\mathbb{Z}}
\renewcommand{\(}{\left(}
\renewcommand{\)}{\right)}

\received[]{}

\title[Eigenvalue asymptotics for a Robin problem]{Asymptotic eigenvalue estimates for a Robin problem with a large parameter}

\author[Abbreviated author(s)'s name(s)]{Pavel Exner, Alexander Minakov, and Leonid Parnovski\thanks{The research has been supported by  the projects ``Support of inter-sectoral mobility and quality enhancement of research teams at Czech Technical University in Prague'', CZ.1.07/2.3.00/30.0034, sponsored by European Social Fund in the Czech Republic, and \mbox{14-06818S} of the Czech Science Foundation. A.M. expresses his gratitude to Stepan Manko and Satoshi Ohya for useful discussions.}}

\begin{document}


\maketitle

\begin{abstract}
Robin problem for the Laplacian in a bounded planar domain with a smooth boundary and a large parameter in the boundary condition is considered. We prove a two-sided three-term asymptotic estimate for the negative eigenvalues. Furthermore, improving the upper bound we get a two term asymptotics in terms of the coupling constant and the maximum of the boundary curvature.
\end{abstract}



\section{Introduction and the main result}

Asymptotic properties of eigenvalues belong among the most often studied problems in the spectral theory. In this paper we are going to discuss an asymptotics of the ``attractive'' Robin problem for the Laplacian in a bounded domain of $\R^2$ in the situation when the parameter $\beta$ in the boundary condition assumes large values. The problem has a physical motivation; it naturally arises in the study of reaction-diffusion equation where a distributed absorbtion competes with a boundary source -- see \cite{Lacey...1996}, \cite{Lacey...1998} for details. At the same the question is of mathematical interest. In a recent paper, Levitin and Parnovski \cite{Levitin_Parnovskii 2008} investigated the asymptotic behavior of the principal eigenvalue and showed that its leading term is $-c\beta^2$ where $c=1$ if the domain boundary is smooth and $c>1$ if it has angles. The same one-term asymptotics is known to hold in the former case also for higher eigenvalues \cite{DK10}.

A related asymptotic problem is encountered in the theory of leaky quantum graphs \cite{Ex08} where the dynamics is not constrained to a bounded region, instead it is governed by a singular Schr\"odinger operator with an attractive interaction supported by a manifold or complex of a lower dimension. A particularly close analogy occurs in the two-dimensional situation when the interaction support is closed smooth loop; using a combination of bracketing and estimates with separated variables, one is able to derive an asymptotic expansion of negative eigenvalues \cite{Exner_Yoshitomi 2002} in which the absolute term with respect to the coupling parameter is given by a one-dimensional Schr\"odinger operator with a potential determined by the geometry of the problem, specifically the curvature of the loop.

This inspires the question whether the technique used for the singular Schr\"o\-dinger operators cannot be used also for Robin ``billiards'' with a smooth boundary. This is the main topic of the present paper. We are going to show first that in distinction to the Schr\"odinger case the method of \cite{Exner_Yoshitomi 2002} does not yield an asymptotic expansion, but two-sided asymptotic estimates only, which squeeze only when the domain is a circular disc. On the other hand, these estimates hold true not only for the principal eigenvalue, and moreover, they have three terms in the powers of $\beta$ which improves, in particular, the result obtained in \cite{Levitin_Parnovskii 2008} for smooth boundaries. On the other hand, the result admits an improvement. Replacing the upper bound by a variational estimate similar to that employed recently by Pankrashkin \cite{Pankrashkin 2013} for the principal eigenvalue, we obtain a bound in which only the maximum of the boundary curvature appears, and as a result, a two-term asymptotic expansion.

Let us now state the problem properly. We suppose that $\Omega$ be an open, simply connected set in $\mathbb{R}^2$ with a closed $C^4$ Jordan boundary $\partial \Omega=\Gamma:[0,L]\ni s\mapsto (\Gamma_1,\Gamma_2)\in\mathbb{R}^2$ which is parametrized by its arc length; for definiteness we choose the clockwise orientation of the boundary. Let $\gamma: [0,L]\rightarrow\mathbb{R}$ be the signed curvature of $\Gamma$, i.e. $\gamma(s)= \Gamma_1''(s)\Gamma_2'(s)-\Gamma_2''(s)\Gamma_1'(s)$. We investigate the spectral boundary-value problem
\begin{eqnarray}
-\Delta f=\lambda f\; &\textrm {in}& \Omega \nonumber \\ [-.5em]
&& \label{problem} \\ [-.5em]
\frac{\partial f}{\partial n}=\beta f\; &\textrm{on}& \Gamma \nonumber
\end{eqnarray}
with a parameter $\beta>0,$ which will be in the following assumed to be large; the symbol $\frac{\partial}{\partial n}$ in (\ref{problem}) denotes the outward normal derivative. It is straightforward to check that the quadratic form
\begin{equation}\label{differential form}
q_{\beta}[f] =\|\nabla f\|^2_{L^2(\Omega)}-\beta\int\limits_{\Gamma}|f(x)|^2\d s
\end{equation}
with $\mathrm{Dom}(q_{\beta}) = H^{1}(\Omega)$ is closed and below bounded; we denote by $H_{\beta}$ the unique self-adjoint operator associated with it. Our main goal is to study the asymptotic behavior of the negative eigenvalues of $H_{\beta}$ as parameter $\beta$ tends to infinity. To state the result, we introduce the one-dimensional Schr\"odinger operator
\begin{equation}\label{S}
S=-\frac{\d^2}{\d s^2}-\frac{1}{4}\gamma^2(s) \quad \textrm{in}\;\;
L^2(0,L)
\end{equation}
with the domain
\begin{equation}\label{P}
P=\left\{f\in H^2(0,L):\ f(0)=f(L),\; f'(0)=f'(L)\right\}.
\end{equation}
We use the symbol $\mu_j$ for the $j$-th eigenvalue of $S$ counted with the multiplicity, $j\in\mathbb{N}$, and furthermore, we denote
$\gamma^*=\max\limits_{[0,L]}\gamma(s)$ and $\gamma_*=\min\limits_{[0,L]}\gamma(s).$

Our main result reads then as follows.

\begin{theorem}\label{main theorem}
Under the stated assumptions, to any fixed integer $n$ there exists a $\beta(n)>0$ such that the number of negative eigenvalues of $H_{\beta}$ is not smaller than~$n.$ For $\beta >\beta(n)$ we denote by $\lambda_n(\beta)$ the $n$-th eigenvalue of $H_{\beta}$ counted with the multiplicity. Then $\lambda_n(\beta)$ satisfies for $\beta\to\infty$ the asymptotic estimates
\begin{equation}\label{estimate_lambda n}
-\(\beta+\frac{\gamma^*}{2}\)^2+\mu_n+ \OO\(\frac{\log\beta}{\beta}\) \le \lambda_n(\beta)
\le -\(\beta+\frac{\gamma_*}{2}\)^2+\mu_n+ \OO\(\frac{\log\beta}{\beta}\).
\end{equation}
\end{theorem}

\begin{remarks}
{\rm (a) It will be clear from the proof that the assumption about simple connectedness of $\Omega$ is done mostly for the sake of simplicity. The result extends easily to multiply connected domains, in general with different parameters at different components of the boundary; each of the components then gives rise to a series of negative eigenvalues tending to $-\infty$ in the limit.

\smallskip

\noindent (b) In the light of the following result the upper bound in \eqref{estimate_lambda n} is not of much use. We include it primarily to illustrate the significant difference between the ``two-sided'' situation discussed in \cite{Exner_Yoshitomi 2002} and the ``one-sided'' one treated here.
}
\end{remarks}

\smallskip

As we have indicated, the upper bound can be improved:

\begin{theorem} \label{thm: better upper}
In the asymptotic regime $\beta\rightarrow+\infty$ the inequality
$$\lambda_n(\beta)\le-\beta^2-\gamma^*\beta+\OO\big(\beta^{2/3}\big)$$
is valid for any fixed $n$. Consequently, the $j$-th eigenvalue behaves asymptotically as
$$\lambda_n(\beta) = -\beta^2-\gamma^*\beta+\OO\big(\beta^{2/3}\big).$$
\end{theorem}

\medskip

\noindent Thus we obtain a two-term asymptotics which, in contrast to the Schr\"odinger operator case treated in \cite{Exner_Yoshitomi 2002}, is not precise enough to distinguish between individual eigenvalues whose mutual distances are expected to be of order of $\OO(1)$.

\section{Proof of Theorem \ref{main theorem}}
Let us first introduce some quadratic forms and operators which we shall need in the argument. To begin with, we need the following result, which is a straightforward modification of Lemma~2.1 of \cite{Exner_Yoshitomi 2002}, hence we skip the proof.

\begin{lemma}
Let $\Phi$ be the map
\[
[0,L)\times(0,a)\ni(s,u)\mapsto(\Gamma_1(s)+u\Gamma_2'(s),\Gamma_2(s)-u\Gamma_1'(s))\in\mathbb{R}^2.
\]
Then there exists an $a_1>0$ such that the map $\Phi$ is injective for any $a\in(0,a_1].$
\end{lemma}

Choose an $a$ satisfying $0<a<a_1$, to be specified later, and let $\Sigma_a$ be the strip neighborhood of $\Gamma$ of width $a$,
\[
\Sigma_a:=\Phi([0,L)\times[0,a)).
\]
Then $\Omega\setminus\Sigma_a=\Lambda_a$ is a compact simply connected domain with the boundary $\Gamma_a:= \Phi([0,L)\times\{a\})$. We define
\begin{eqnarray*}
&& q_{a,\beta}^D[f]:=\|\nabla f\|^2_{\Sigma_a}-\beta\int\limits_{\Gamma}|f(x)|^2\d s\quad\textrm{ for }\: f\in\left\{f\in H^1(\Sigma_a): f|_{\Gamma_a}=0\right\}, \\ &&
q_{a,\beta}^N[f]:=\|\nabla f\|^2_{\Sigma_a}-\beta\int\limits_{\Gamma}|f(x)|^2\d s\quad\textrm{ for }\: f\in H^1(\Sigma_a),
\end{eqnarray*}
and denote by $L_{a,\beta}^D$ and $L_{a,\beta}^N$ the self-adjoint operators associated with the forms $q_{a,\beta}^D$ and $q_{a,\beta}^N,$ respectively. The first key component of the proof is to use the Dirichlet-Neumann bracketing -- see \cite[Sec.~XIII.15, Prop.~4]{Reed Simon 4 1978} -- imposing additional boundary condition at $\Gamma_a$. This yields
\begin{equation}\label{inequalities H_beta}
(-\Delta^N_{\Lambda_a})\oplus L_{a,\beta}^N\leq H_{\beta}\leq
(-\Delta^D_{\Lambda_a})\oplus L_{a,\beta}^D
\end{equation}
in $L^2(\Omega)= L^2(\Lambda_a)\oplus L^2(\Sigma_a)$ where the inequality should be understood, of course, in the variational sense. Since the estimating operators have the direct-sum structure and the first terms in the inequalities (\ref{inequalities H_beta}) referring to the part of $\Omega$ separated from the boundary are positive, in order to estimate the negative eigenvalues of $H_{\beta}$ it is sufficient to estimate those of $L^D_{a,\beta},$ $L^N_{a,\beta}$.

To this aim we use the second main trick, introducing a ``straightening'' transformation in the spirit of \cite{Exner_Seba 1989}, to produce a pair of
operators in $L^2((0,L)\times(0,a))$ that are unitarily equivalent to $L_{a,\beta}^D$ and $L_{a,\beta}^N,$ respectively. Specifically, we introduce the following change of variables,
$$
f(x_1,x_2)=\frac{1}{(1-u\gamma(s))^{1/2}}\ \varphi(s,u);
$$
then it is straightforward to check that for any function $f\in H^2(\Sigma_a)$ we have also $\varphi\in H^2((0,L)\times(0,a))$ and
\begin{eqnarray*}
\lefteqn{|f_{x_1}|^2+|f_{x_2}|^2 =\left[\frac{1}{(1-u\gamma(s))^2}\left|\frac{\partial\varphi}{\partial s}\right|^2+
\left|\frac{\partial\varphi}{\partial u}\right|^2 +\widetilde
V(s,u)|\varphi|^2\right.} \\ && \left.+\frac{u\gamma'(s)}{2(1-u\gamma(s))^3}\(\varphi\
\overline{\frac{\partial\varphi} {\partial s}}+\overline{\varphi}\
\frac{\partial\varphi} {\partial
s}\)+\frac{\gamma(s)}{2(1-u\gamma(s))}\(\varphi\
\overline{\frac{\partial\varphi} {\partial u}}+\overline{\varphi}\
\frac{\partial\varphi} {\partial u}\)\right]
\end{eqnarray*}
with
\[
\widetilde V(s,u)=\frac{\gamma^2(s)}{4(1-u\gamma(s))^2}+\frac{u^2(\gamma'(s))^2}{4(1-u\gamma(s))^4}\,,
\]
where we employ the usual shorthands, $f_{x_j}= \frac{\partial f}{\partial x_j}$, and furthermore
\begin{eqnarray*}
&& \qquad\qquad \iint\limits_{\Sigma_a}\(|f_{x_1}|^2+|f_{x_2}|^2\)\d x_1\d
x_2-\beta\int\limits_{\Gamma}|f(x)|^2\d s \\ && =\iint\limits_{0\
0}^{\ \ \ a\ L}
\frac{1}{(1-u\gamma(s))^2}\left|\frac{\partial\varphi}{\partial
s}\right|^2\d s\,\d u+ \iint\limits_{0\ 0}^{\ \ \ a\ L}
\left|\frac{\partial\varphi}{\partial u}\right|^2\d s\,\d u +
\iint\limits_{0\ 0}^{\ \ \ a\ L}V(s,u) \left|\varphi\right|^2\d
s\,\d u \\ && \qquad -\int\limits_0^L\(\frac{\gamma(s)}{2}+\beta\)|\varphi(s,0)|^2\d
s+\int\limits_0^L\frac{\gamma(s)}{2(1-a\gamma(s))}|\varphi(s,a)|^2\d s,
\end{eqnarray*}
where
\begin{eqnarray*}
\lefteqn{V(s,u)=\widetilde
V(s,u)-\frac{\partial}{\partial
s}\(\frac{u\gamma'(s)}{2(1-u\gamma(s))^3}\)-\frac{\partial}{\partial
u}\(\frac{\gamma(s)}{2(1-u\gamma(s))}\)} \\ &&
=-\dsfrac{\gamma^2(s)}{4(1-u\gamma(s))^2}-\frac{u\gamma''(s)}{2(1-u\gamma(s))^3}
-\frac{5}{4}\frac{u^2(\gamma'(s))^2}{(1-u\gamma(s))^4}.
\end{eqnarray*}
Armed with these formul{\ae} we can now introduce the two operators in $L^2((0,L)\times(0,a))$ unitarily equivalent to $L_{a,\beta}^D$ and $L_{a,\beta}^N,$ respectively. On the domains
\[
Q_a^D=\left\{\varphi\in H^1((0,L)\times(0,a)):\, \varphi(L,.)=\varphi(0,.)\:\textrm{ on } (0,a),\:
\varphi(.,a)=0\:\textrm{ on } (0,L)\right\}
\]
and
\[
Q_a^N=\left\{\varphi\in H^1((0,L)\times(0,a)):\quad \varphi(L,.)=\varphi(0,.)\:\textrm{ on } (0,a)\right\},
\]
we define the quadratic forms
\begin{eqnarray}\nonumber
\lefteqn{b_{a,\beta}^D[\varphi]= \iint\limits_{0\ 0}^{\ \ \ a\ L}
\frac{1}{(1-u\gamma(s))^2}\left|\frac{\partial\varphi}{\partial
s}\right|^2\d s\,\d u+ \iint\limits_{0\ 0}^{\ \ \ a\ L}
\left|\frac{\partial\varphi}{\partial u}\right|^2\d s\,\d u} \\ &&
+\iint\limits_{0\ 0}^{\ \ \ a\ L}V(s,u) \left|\varphi\right|^2\d
s\,\d
u-\int\limits_0^L\(\frac{\gamma(s)}{2}+\beta\)|\varphi(s,0)|^2\d s
\label{bD}\end{eqnarray}
and
\begin{eqnarray*}
\lefteqn{b_{a,\beta}^N[\varphi]= \iint\limits_{0\ 0}^{\ \ \ a\
L} \frac{1}{(1-u\gamma(s))^2}\left|\frac{\partial\varphi}{\partial
s}\right|^2\d s\,\d u+ \iint\limits_{0\ 0}^{\ \ \ a\ L}
\left|\frac{\partial\varphi}{\partial u}\right|^2\d s\,\d u} \\ &&
+ \iint\limits_{0\ 0}^{\ \ \ a\ L}V(s,u) \left|\varphi\right|^2\d s\,\d u
-\int\limits_0^L\(\frac{\gamma(s)}{2}+\beta\)|\varphi(s,0)|^2\d s \\ &&
+\int\limits_0^L\frac{\gamma(s)}{2(1-a\gamma(s))}|\varphi(s,a)|^2\d s,
\end{eqnarray*}
respectively. It is easy to check the following claim analogous to Lemma~2.2
of \cite{Exner_Yoshitomi 2002}.

\begin{lemma} \label{l: uniteq}
The operators $B_{a,\beta}^D$ and $B_{a,\beta}^N$ associated with above quadratic forms are unitarily
equivalent to $L_{a,\beta}^D$ and $L_{a,\beta}^N,$ respectively.
\end{lemma}

In the next step we estimate $B_{a,\beta}^D$ and $B_{a,\beta}^N$ just introduced by operators with separated variables. We put\footnote{There is a typo in \cite{Exner_Yoshitomi 2002}; the second term in the definition of $V_+$ there has to be deleted.}
\begin{eqnarray*}
&& \gamma_+=\max\limits_{[0,L]}|\gamma(.)|, \quad \gamma'_+=\max\limits_{[0,L]}|\gamma'(.)|,\quad \gamma''_+=\max\limits_{[0,L]}|\gamma''(.)|, \\
&& V_+(s)=\frac{-\gamma^2(s)}{4(1+a\gamma_+)^2}+\frac{a\gamma''_+}{2(1-a\gamma_+)^3}\ , \\
&& V_-(s)=\frac{-\gamma^2(s)}{4(1-a\gamma_+)^2}-\frac{a\gamma''_+}{2(1-a\gamma_+)^3}-\frac{5}{4}\frac{a^2(\gamma'_+)^2}{(1-a\gamma_+)^4}.
\end{eqnarray*}
For an $a$ satisfying $0<a<\gamma_+/2$ and $\varphi$ belonging to $Q_a^D$ and $Q_a^N$, respectively, we define
\begin{eqnarray*}
\lefteqn{\widetilde b_{a,\beta}^D[\varphi] =
(1-a\gamma_+)^{-2}\iint\limits_{0\ 0}^{\ \ \ a\ L}
\left|\frac{\partial\varphi}{\partial s}\right|^2\d s\,\d u+
\iint\limits_{0\ 0}^{\ \ \ a\ L}
\left|\frac{\partial\varphi}{\partial u}\right|^2\d s\,\d u} \\ && +
\iint\limits_{0\ 0}^{\ \ \ a\ L}V_+(s) \left|\varphi\right|^2\d
s\,\d u-\(\frac{\gamma_*}{2}+\beta\)\int\limits_0^L|\varphi(s,0)|^2\d s
\end{eqnarray*}
and
\begin{eqnarray*}
\lefteqn{\widetilde b_{a,\beta}^N[\varphi]=
(1+a\gamma_+)^{-2}\iint\limits_{0\ 0}^{\ \ \ a\ L}
\left|\frac{\partial\varphi}{\partial s}\right|^2\d s\,\d u+
\iint\limits_{0\ 0}^{\ \ \ a\ L}
\left|\frac{\partial\varphi}{\partial u}\right|^2\d s\,\d u} \\ && +
\iint\limits_{0\ 0}^{\ \ \ a\ L}V_-(s) \left|\varphi\right|^2\d
s\,\d u-\(\frac{\gamma^*}{2}+\beta\)\int\limits_0^L|\varphi(s,0)|^2\d
s-\frac{\gamma_+}{2(1-a\gamma_+)}\int\limits_0^L|\varphi(s,a)|^2\d s.
\end{eqnarray*}
Then we have
\begin{equation}\label{inequality b^D}
b_{a,\beta}^D[\varphi]\leq \widetilde
b_{a,\beta}^D[\varphi]\quad \textrm{ for }\,f\in Q_a^D,
\end{equation}
\begin{equation}\label{inequality b^N}
b_{a,\beta}^N[\varphi]\geq \widetilde
b_{a,\beta}^N[\varphi]\quad \textrm{ for }\,f\in Q_a^N.
\end{equation}
Let $\widetilde H_{a,\beta}^D$ and $\widetilde H_{a,\beta}^N$ be the self-adjoint operators associated with the forms $\widetilde b_{a,\beta}^D$ and $\widetilde b_{a,\beta}^N$, respectively. By $T_{a,\beta}^D$ we denote the self-adjoint operator associated with the form
\[
t_{a,\beta}^D[\varphi]=\int\limits_0^a|\varphi'(u)|^2\d u-\(\frac{\gamma_*}{2}+\beta\)|\varphi(0)|^2
\]
defined on $\{\varphi\in H^1(0,a):\: \varphi(a)=0\}$. Similarly, $T_{a,\beta}^N$ is the self-adjoint operator associated with the form
\[
t_{a,\beta}^N(\varphi,\varphi)=\int\limits_0^a|\varphi'(u)|^2\d u-\(\frac{\gamma^*}{2}+\beta\)|\varphi(0)|^2, \quad \varphi\in  H^1(0,a).
\]
Furthermore, we introduce the operators
\[
U_a^D=(1-a\gamma_+)^{-2}\(-\,\frac{\d^2}{\d s^2}\)+V_+(s)\,, \quad U_a^N=(1+a\gamma_+)^{-2}\(-\,\frac{\d^2}{\d s^2}\)+V_-(s)
\]
in $L^2(0,L)$, the domain of both of them being $P$ given by (\ref{P}). Then we have
\begin{equation}\label{H tilde as sum}
\widetilde H_{a,\beta}^D=U_a^D\otimes I+I\otimes T_{a,\beta}^D,
\qquad \widetilde H_{a,\beta}^N=U_a^N\otimes I+I\otimes
T_{a,\beta}^N,
\end{equation}
and we can estimate contributions from the longitudinal and transverse variables separately. What concerns the former, we denote by $\mu_j^D(a),$ $\mu_j^N(a)$ the $j$-th eigenvalue of $U_a^D,$ $U_a^N,$ respectively, counted with the multiplicity, and use Proposition~2.3 of \cite{Exner_Yoshitomi 2002} which contains the following claim:

\begin{lemma} \label{longlemma}
There exists a constant $C>0$ such that the estimates
\begin{equation}\label{estimate mu_j^D}
|\mu_j^D(a)-\mu_j|\leq C aj^2
\end{equation}
and
\begin{equation}\label{estimate mu_j^N}
|\mu_j^N(a)-\mu_j|\leq C aj^2
\end{equation}
hold for any $j\in\mathbb{N}$ and $0<a<1/(2\gamma_+).$ where $C$ is
independent on $j,a.$
\end{lemma}

\noindent We stress that the constant $C$ here is independent of $j$ and $a$. As for the transverse part, let us estimate first the principal eigenvalue of $T_{a,\beta}^D$.

\begin{lemma}\label{lemma T D eigenvalue}
Assume that
$a\(\beta+\frac{\gamma_*}{2}\)>\dsfrac{4}{3}.$ Then
$T_{a,\beta}^D$ has only one negative eigenvalue which we denote
by $\zeta_{a,\beta}^D.$ It satisfies the inequalities
\[
-\(\beta+\frac{\gamma_*}{2}\)^2\leq\zeta_{a,\beta}^D\leq -\(\beta+\frac{\gamma_*}{2}\)^2+4\(\beta+\frac{\gamma_*}{2}\)^2
\e^{-a\(\beta+\frac{\gamma_*}{2}\)}.
\]
\end{lemma}
\begin{proof}
Notice that the domain of the operator is
$$
D(T_{a,\beta}^D)=\left\{\varphi\in H^2(0,a):\
\varphi'(0)=-\(\frac{\gamma_*}{2}+\beta\)\varphi(0),\ \varphi(a)=0
\right\}.
$$
Assume that $-k^2$ with $k>0$ is an eigenvalue of $T_{a,\beta}^D,$ and let a nonzero $\varphi$ be the corresponding eigenfunction, then we have
\begin{enumerate}
\item $-\varphi''(u)=-k^2\varphi(u);$
\item $\varphi'(0)=-\(\frac{\gamma_*}{2}+\beta\)\varphi(0);$
\item $\varphi(a)=0.$
\end{enumerate}
In view of the first property, the eigenfunction $\varphi$ is of the form
\[
\varphi(u)=A\e^{ku}+B\e^{-ku}.
\]
Furthermore, the requirements (2) and (3) yield $kA-kB=\(-\frac{\gamma_*}{2}-\beta\)(A+B)$ and  $A\e^{ka}+B\e^{-ka}=0,$ respectively. Thus the coefficients $A,\,B$ have to satisfy the equation
\[
\begin{pmatrix}\e^{ka}&\e^{-ka}\\k+\frac{\gamma_*}{2}+\beta&-(k-\frac{\gamma_*}{2}-\beta)\end{pmatrix}
\begin{pmatrix}A\\B\end{pmatrix}=0.
\]
Since $(A,B)\neq(0,0),$ we get
\[
\det\begin{pmatrix}\e^{ka}&\e^{-ka}\\k+\frac{\gamma_*}{2}+\beta&-(k-\frac{\gamma_*}{2}-\beta)\end{pmatrix}=0
\]
which is equivalent to $g_{a,\beta}(k):=2ak+\log\(\beta+\frac{\gamma_*}{2}-k\)-\log\(\beta+\frac{\gamma_*}{2}+k\)=0.$ It is easy to see that also the converse is true: if $g_{a,\beta}(k)=0$, then $-k^2$ is an eigenvalue of $T_{a,\beta}^D.$ Let us now show that $g_{a,\beta}(.)$ has a unique zero in
$\(0,\beta+\frac{\gamma_*}{2}\)$. By definition we have $g_{a,\beta}(0)=0,$ and since
\[
\frac{\d g_{a,\beta}(k)}{\d k}=\frac{2a\(\beta+\frac{\gamma_*}{2}\)^2-2\(\beta+\frac{\gamma_*}{2}\)-2ak^2}{\(\beta+\frac{\gamma_*}{2}\)^2-k^2}
\]
we can claim that $g_{a,\beta}$ is monotonically increasing in $\(0,\ \beta+\dsfrac{\gamma_*}{2}-\dsfrac{1}{a}\)$ and it is monotonically decreasing in $\(\beta+\dsfrac{\gamma_*}{2}-\dsfrac{1}{a}\,,\ \beta+\dsfrac{\gamma_*}{2}\).$ Moreover, we have
\[
\lim\limits_{k\rightarrow \beta+\frac{\gamma_*}{2}}g_{a,\beta}=-\infty\,;
\]
this implies that the function $g_{a,\beta}$ has a unique zero in $\(0,\beta+\dsfrac{\gamma_*}{2}\).$ Moreover, since $a(\beta+\dsfrac{\gamma_*}{2})> \frac{4}{3},$ we have $\sqrt{\(\beta+\dsfrac{\gamma_*}{2}\)\(\beta+\dsfrac{\gamma_*}{2}-\dsfrac{1}{a}\)}>\dsfrac{1}{2}\(\beta+\dsfrac{\gamma_*}{2}\).$
Consequently, the solution $k$ has the form $k=\beta+\dsfrac{\gamma_*}{2}-s,$ $0<s<\frac{1}{2}\(\beta+\frac{\gamma_*}{2}\).$ Taking into account
the relation $g_{a,\beta}(k)=0,$ we get
\[
\log s=\log(2\beta+\gamma_*-s)-2a\(\beta+\frac{\gamma_*}{2}-s\)\leq\log(2\beta+\gamma_*)-a\(\beta+\frac{\gamma_*}{2}\).
\]
Hence we obtain $s\leq(2\beta+\gamma_*)\e^{-a\(\beta+\gamma_*/2\)}\,$ which concludes the proof.
\end{proof}

Next we estimate the first eigenvalue of $T_{a,\beta}^N.$

\begin{lemma}\label{lemma T N eigenvalue}
Assume that
$\(\beta+\frac{\gamma^*}{2}\)>\max\left\{\frac{\gamma_+}{2(1-a\gamma_+)},\ \frac{2\log 5}{3a}\right\}.$ Then $T_{a,\beta}^N$ has a unique negative eigenvalue $\zeta_{a,\beta}^N,$ and moreover, we have
\[
-\(\beta+\frac{\gamma^*}{2}\)^2-\frac{45}{4}\(\beta+\frac{\gamma^*}{2}\)^2
\e^{-a\(\beta+\frac{\gamma^*}{2}\)}\leq\zeta_{a,\beta}^N\leq
-\(\beta+\frac{\gamma^*}{2}\)^2.
\]
\end{lemma}
\begin{proof}
The operator domain in this case looks as follows,
$$
D(T_{a,\beta}^N)=\left\{\varphi\in H^2(0,a):\
\varphi'(0)=-\(\frac{\gamma^*}{2}+\beta\)\varphi(0),\
\varphi'(a)=\frac{\gamma_+}{2(1-a\gamma_+)}\varphi(a) \right\}.
$$
Assume again that $-k^2$ with $k>0$ is an eigenvalue of $T_{a,\beta}^N$ corresponding to a nonzero eigenfunction $\varphi$. As in the proof of Lemma
\ref{lemma T D eigenvalue} we infer that $-k^2$ is an eigenvalue of $T_{a,\beta}^N$ if and only if\footnote{There is a misplaced exponential in the analogous proof in \cite{Exner_Yoshitomi 2002} which does not affect the claim.}
\begin{equation}\label{equation for T N eigenvalue}
\e^{2ka}=\frac{k+\frac{\gamma^*}{2}+\beta\
}{k-\frac{\gamma^*}{2}-\beta\ }\ \cdot\
\frac{k+\frac{\gamma_+}{2(1-a\gamma_+)}\
}{k-\frac{\gamma_+}{2(1-a\gamma_+)}}\ .
\end{equation}
Since the left-hand side of the last equation is strictly increasing, and the right-hand side is strictly decreasing for $k>0$, then the equation (\ref{equation for T N eigenvalue}) has a unique positive solution which lies in fact in the subinterval $\(\beta+\frac{\gamma^*}{2},+\infty\).$

Next we will show that (\ref{equation for T N eigenvalue}) has no solutions in the interval $k\geq\frac{3}{2}\(\beta+\frac{\gamma^*}{2}\).$ Suppose that the opposite is true. As $\frac{\gamma_+}{2(1-a\gamma_+)}<\beta+\frac{\gamma^*}{2},$ we have
\[
\frac{k+\frac{\gamma^*}{2}+\beta\ }{k-\frac{\gamma^*}{2}-\beta\
}\ \cdot\ \frac{k+\frac{\gamma_+}{2(1-a\gamma_+)}\
}{k-\frac{\gamma_+}{2(1-a\gamma_+)}}\leq
\(\frac{k+\frac{\gamma^*}{2}+\beta\ }{k-\frac{\gamma^*}{2}-\beta\
}\)^2.
\]
However, since we assume $k\geq\frac{3}{2}\(\beta+\frac{\gamma^*}{2}\),$ this would imply
 \[
\e^{2ka} \leq
\(\frac{\frac{3}{2}\(\beta+\frac{\gamma^*}{2}\)+\frac{\gamma^*}{2}+\beta\
}{\frac{3}{2}\(\beta+\frac{\gamma^*}{2}\)-\frac{\gamma^*}{2}-\beta\
}\)^2=25.
\]
On the other hand, we have $\e^{2ka}\geq\e^{3a\(\beta+\frac{\gamma^*}{2}\)}>25,$ so we come to a contradiction. Hence the solution $k$ of (\ref{equation for T N eigenvalue}) is of the form $k=\beta+\frac{\gamma^*}{2}+s$ with $0<s<\frac{1}{2}\(\beta+\frac{\gamma^*}{2}\),$ and using (\ref{equation
for T N eigenvalue}) once again we get
\[
\e^{2ka}\leq
\(\frac{k+\frac{\gamma^*}{2}+\beta\ }{k-\frac{\gamma^*}{2}-\beta\
}\)^2\leq \(\frac{2\beta+\gamma^*+s}{s }\)^2\leq
\(\frac{\frac{5}{2}\(\beta+\frac{\gamma^*}{2}\)}{s}\)^2,
\]
which further implies
\[
s\leq\frac{5}{2}\(\beta+\frac{\gamma^*}{2}\)\e^{-ka}=
\frac{5}{2}\(\beta+\frac{\gamma^*}{2}\)\e^{-a\(\beta+\frac{\gamma^*}{2}\)-sa}\leq
\frac{5}{2}\(\beta+\frac{\gamma^*}{2}\)\e^{-a\(\beta+\frac{\gamma^*}{2}\)}.
\]
This completes the proof of Lemma \ref{lemma T N eigenvalue}.
\end{proof}

Now we are finally in position to \textsl{prove Theorem \ref{main theorem}.}  We use first the bracketing to squeeze the eigenvalues in question between those of the operators \eqref{H tilde as sum}. Since the latter have separated variables, their eigenvalues are sums of eigenvalues of the longitudinal and transverse component which we have estimated in Lemmata \ref{longlemma} and \ref{lemma T D eigenvalue}, \ref{lemma T N eigenvalue}, respectively, and it is sufficient to choose $a= \frac{6}{\beta} \log\beta$ to get \eqref{estimate_lambda n}. $\hfill\Box$

\medskip

Note that while the argument is pretty much the same as in the proof of Theorem~1 of \cite{Exner_Yoshitomi 2002}, with Propositions 2.3--2.5 there replaced by the above mentioned lemmata, the result is much weaker due to the presence of the last term in the form $b^D_{a,\beta}$  and its counterpart in $b^N_{a,\beta}$. In particular, the estimates of Theorem~ \ref{main theorem} squeeze to produce an exact asymptotic expansion if and only if the curvature is constant, $\gamma^*=\gamma_*$. Let us now look at this case in more detail:

\begin{example}
Let $\Omega$ be a disc of radius $R$ centered at the origin. In this case we have
\begin{equation}\label{gamma explicit}
\gamma(s)\equiv\gamma=\frac{1}{R}
\end{equation}
and the  eigenvalues $\mu_j$ of the comparison operator $S$ given by (\ref{S}) can be computed explicitly,
\begin{equation}\label{mu_j explicit}
\mu_j=\(-\ \frac{1}{4}+\left[\frac{j}{2}\right]^2\)R^{-2},
\end{equation}
where $[y]$ denotes the maximum integer which less or equal to $y.$ We introduce the usual polar coordinates,
\[
\left\{\begin{array}{ccc}x=r\cos
\theta\\y=r\sin\theta\end{array}\right. \qquad0\leq r\leq R\,,\
0\leq\theta<2\pi\,,
\]
writing with an abuse of notation $f(x,y)\equiv f(r,\theta).$ Equations (\ref{problem}) with $\lambda=-k^2$ now read
\begin{equation}\label{Laplacian polar coordinates and boundary condition}
\left\{\begin{array}{l}\dsfrac{\partial^2f}{\partial
r^2}+\frac{1}{r}\frac{\partial f}{\partial r}+ \frac{1}{r^2}\
\frac{\partial^2f}{\partial\theta^2}=k^2f,\\\\\left.\dsfrac{\partial
f}{\partial r}\right|_{\ r=R}=\beta
f.\end{array}\right.
\end{equation}
Solution to the first equation in (\ref{Laplacian polar coordinates and boundary condition}) is conventionally sought in the form
\begin{equation}\label{f as a sum polar coordinates}
f(r,\theta)=\sum\limits_{m\in\mathbb{Z}}c_m I_m(kr)\e^{\i
m\theta}.
\end{equation}
Furthermore, the Hamiltonian commutes with the angular momentum operator, $-i \frac{\partial}{\partial\theta}$ with periodic boundary conditions, hence the two operators have common eigenspaces, and we can consider sequence $\{c_m\}$ with nonzero $c_m$ corresponding to a single values of $|m|$; it goes without saying that the discrete spectrum has multiplicity two except the eigenvalue corresponding to $m=0$ which is simple. The boundary condition in (\ref{Laplacian polar coordinates and boundary condition}) can be then rewritten as
\begin{equation}\label{equation I_m (1)}
k I'_m(kR)-\beta I_m(kR)=0.
\end{equation}
for a fixed $m\in\Z$. To find its solutions, let us change the variables to $X=kR,$ $\alpha=\beta R,$ in which case the condition  (\ref{equation I_m (1)}) reads
\begin{equation}\label{equation I_m (2)}
\frac{X I'_m(X)}{I_m(X)}=\alpha.
\end{equation}
The function at the left-hand side of (\ref{equation I_m (2)}) is strictly increasing for $k>0,$ hence (\ref{equation I_m (2)}) has a unique solution for any fixed $\alpha$ and $m$. As $\alpha\rightarrow+\infty,$ so does $X$ in (\ref{equation I_m
(2)}), and using the well-known asymptotics of modified Bessel functions, we find
$$
\frac{XI'_m(X)}{I_m(X)}=X-\frac{1}{2}+\frac{4m^2-1}{8X}+O(X^{-2}),\quad
X\rightarrow+\infty.
$$
In combination with the spectral condition (\ref{equation I_m (2)}) this yields
$$
X=\alpha+\frac{1}{2}-\frac{4m^2-1}{8\alpha}+O(\alpha^{-2}),\quad\alpha\rightarrow+\infty.
$$
This, in turn, implies the asymptotics for $X^2$, and returning to the original variables $\beta, k$ we find
$$
-k^2= -\(\beta+\frac{1}{2R}\)^2+\(m^2-\frac{1}{4}\)R^{-2}+\OO(\beta^{-1}),\quad\beta\rightarrow+\infty.
$$
This agrees, of course, with the conclusion of Theorem~\ref{main theorem} according to (\ref{gamma explicit}) and (\ref{mu_j explicit}). At the same time it shows that there is not much room for improving the error term in the theorem, because it differs from the one in this explicitly solvable example by the logarithmic factor only.
\end{example}

\section{Proof of Theorem \ref{thm: better upper}}

The idea is to replace the crude estimate of $B_{a,\beta}^D$ from Lemma~\ref{l: uniteq} by the first operator in (\ref{H tilde as sum}) by a finer one.
Consider first the principal eigenvalue which satisfies $\lambda_1(\beta)\le b_{a,\beta}^D[\varphi]$ for any $\varphi\in Q_a^D$ and choose the following family of trial functions,
$$\hat\varphi(s,u)=\chi_{\varepsilon}(s)\(\e^{-\alpha
u}-\e^{-2a\alpha+u\alpha}\),$$
where $\chi_{\varepsilon}$ is a smooth function on $[0,L]$ with the support located in an $\varepsilon$-neighborhood of a point $s^*$ in which the curvature reaches its maximum, $\gamma(s^*)=\gamma^*$, and $\varepsilon$ is a parameter to be determined later. In view of the boundary compactness and smoothness, at least one such point exists; without loss of generality we may assume that $(s^*-\varepsilon,s^*+\varepsilon)\subset(0,L)$. We shall consider functions of the form
\[\chi_{\varepsilon}(s):=\chi\(\dsfrac{s-s^*+\varepsilon}{2\varepsilon}\),\]
where $\chi(x)$ is a fixed smooth function on $\mathbb{R}$ with the support in the interval $(0,1)$; then we have
\begin{equation}\label{norm chi}
\|\chi_{\varepsilon}\|_{L^2(0,L)}^2=2\varepsilon\|\chi\|_{L^2(0,1)}^2,\quad
\|\chi'_{\varepsilon}\|_{L^2(0,L)}^2=\(2\varepsilon\)^{-1}\|\chi'\|_{L^2(0,1)}^2.
\end{equation}
We also note that on the support of $\chi_{\varepsilon}$, i.e. for any $s\in(s^*-\varepsilon,s^*+\varepsilon)$ we have
\[|\gamma(s)-\gamma^*|<\gamma'_+ |s-s^*| < \gamma'_+ \varepsilon.\]
Computing the terms of the form $b_{a,\beta}^D[\varphi]$ we get for the longitudinal kinetic contribution the estimate
\begin{eqnarray*}
\lefteqn{\hspace{-1.5em} \iint\limits_{0\ 0}^{\ \ \ a\ L}
\frac{1}{(1-u\gamma(s))^2}\left|\frac{\partial\hat\varphi}{\partial
s}\right|^2\d s\,\d u \le \int\limits_{0}^a\int\limits_{
s^*-\varepsilon}^{s^*+\varepsilon}
\(\frac{1}{(1-u\gamma^*)^2}+C\varepsilon u\)\(\e^{-\alpha
u}-\e^{-2a\alpha+u\alpha}\)^2} \\ &&  \times (\chi'(s))^2\d s\,\d u
\left[\(\dsfrac{1}{2\alpha}+\OO\big(\alpha^{-2}\)\big)+C\varepsilon\(\dsfrac{1}{4\alpha^2}+
\OO\big(\alpha^{-3}\)\big)\right]\|\chi\|^2_{L^2(0,L)},
\end{eqnarray*}
where $C>0$ is a generic constant independent of $\beta, a$, and $\varepsilon.$ Similarly,
\[\iint\limits_{0\ 0}^{\ \ \ a\ L}
\left|\frac{\partial\hat\varphi}{\partial u}\right|^2\d s\,\d
u=\dsfrac{\alpha}{2}\(1+\OO\big(\alpha\e^{-2a\alpha}\big)\)\|\chi\|_{L^2(0,L)}^2\]
holds for the transverse kinetic term,
\begin{eqnarray*}
\lefteqn{\int\limits_{0}^a\int\limits_{ 0}^{L}V(s,u) \left|\hat\varphi\right|^2\d
s\,\d u \le \int\limits_{0}^a\int\limits_{
s^*-\varepsilon}^{s^*+\varepsilon}\(-\dsfrac{\(\gamma^*\)^2}{4(1-u\gamma^*)}
-\dsfrac{u\gamma''(s^*)}{2(1-u\gamma^*)^3}-\dsfrac{5}{4}\dsfrac{u^2\(\gamma'(s^*)\)^2}
{\(1-u\gamma^*\)^4}+C\varepsilon\)} \\ && \hspace{-1em} \times \(\e^{-\alpha
u}-\e^{-2a\alpha+u\alpha}\)^2 |\chi(s)|^2\d s\d u
=\(\dsfrac{-\(\gamma^*\)^2}{4}+C\varepsilon\)
\dsfrac{1}{2\alpha}\(1+\OO\big(\alpha^{-1}\big)\)\|\chi\|^2_{L^2(0,L)}
\end{eqnarray*}
for the potential one, and
\[-\int\limits_0^L\(\frac{\gamma(s)}{2}+\beta\)|\hat\varphi(s,0)|^2\d s\geq-\(\beta+\dsfrac{\gamma^*-\varepsilon}{2}\)
\(1-\e^{-2a\alpha}\)^2\|\chi\|_{L^2(0,L)}^2,\]
for the boundary one. Finally, the trial function norm satisfies
\[\iint\limits_{0\ 0}^{\ \ \ a\ L}|\hat\varphi(s,u)|^2\d s\d u=\dsfrac{1}{2\alpha}\(1+\OO\big(\alpha\e^{-2a\alpha}\big)\)\|\chi\|^2_{L^2(0,L)}.\]
Putting these expressions together and and taking (\ref{norm chi}) into account we get
\begin{eqnarray*}
\lefteqn{\dsfrac{b_{a,\beta}^D[\hat\varphi]}{\|\hat\varphi\|^2_{L^2(0,L)}}\le
\dsfrac{1}{4\varepsilon^2}\
\dsfrac{\|\chi'\|_{L^2(0,1)}^2}{\|\chi\|_{L^2(0,1)}^2}\(1+\OO\big(\alpha^{-1}\big)+C\varepsilon\(\dsfrac{1}{2\alpha}+\mathrm{\alpha^{-2}}\)\)} \\ &&
+\alpha^2\(1+\OO\big(\alpha\e^{-2a\alpha}\big)\) +
\(\dsfrac{-\(\gamma^*\)^2}{4}+C\varepsilon\)\(1+\OO\big(\alpha^{-1}\big)\) \\ &&
-2\alpha\(\beta+\dsfrac{\gamma^*-\varepsilon}{2}\)\(1+\OO\big(\alpha\e^{-2a\alpha}\big)\).
\end{eqnarray*}
Now we choose $\alpha=\beta+\dsfrac{\gamma^*}{2}$ in which case the right-hand side of the last inequality becomes
\begin{eqnarray*}
\lefteqn{\hspace{-2em} \dsfrac{1}{4\varepsilon^2}\
\dsfrac{\|\chi'\|_{L^2(0,1)}^2}{\|\chi\|_{L^2(0,1)}^2}\(1+\OO\big(\beta^{-1}\big)+C\varepsilon\(\dsfrac{1}{2\beta}+\mathrm{\beta^{-2}}\)\)
-\(\beta+\dsfrac{\gamma^*}{2}\)^2} \\ && +\varepsilon\(\beta+\dsfrac{\gamma^*}{2}\)
+\(\dsfrac{-\(\gamma^*\)^2}{4}+C\varepsilon\)\(1+\OO\big(\alpha^{-1}\big)\),
\end{eqnarray*}
and to optimize the last formula with respect to $\varepsilon$ we take $\varepsilon=\beta^{-1/3},$ which yields the estimate
\begin{equation} \label{finalbound}
\dsfrac{b_{a,\beta}^D[\hat\varphi]}{\|\hat\varphi\|^2_{L^2(0,L)}}\le-\(\beta+\dsfrac{\gamma^*}{2}\)^2+\OO\big(\beta^{2/3}\big)
\end{equation}
proving the result. The argument for the higher eigenfunctions proceeds in the same way. We employ trial functions of the form
$$\hat\varphi_j(s,u)=\chi_{\varepsilon,j}(s)\(\e^{-\alpha u}-\e^{-2a\alpha+u\alpha}\),$$
where the longitudinal part is constructed from a shifted function $\chi$, for instance
\[\chi_{\varepsilon,j}(s):=\chi\(\dsfrac{s-s^*+(2j-1)\varepsilon}{2\varepsilon}\).\]
The above estimate of the form remains essentially the same, up to the values of the constants involved. By construction, the functions $\chi_{\varepsilon,j}$ with different values of~$j$ have disjoint supports, hence $\hat\varphi_j$ is orthogonal to $\hat\varphi_i,\: i=1,\dots,j-1,$ and by the min-max principle \cite[Sec.~XIII.1]{Reed Simon 4 1978} the eigenvalue $\lambda_j(\beta)$ has again the upper bound given by the right-hand side of (\ref{finalbound}). \quad $\Box$

\end{document}